\documentclass{article}
\usepackage{iclr2025_conference,times}

\usepackage{microtype}
\usepackage{graphicx}
\usepackage{subfigure}
\usepackage{booktabs} 
\usepackage{float}
\usepackage{mathrsfs}
\usepackage{array}
\usepackage{booktabs}
\usepackage[margin=1in]{geometry}

\usepackage{hyperref}

\iclrfinalcopy

\newcommand{\zzy}[1]{\textcolor{black}{#1}}

\usepackage{amsmath}
\usepackage{amssymb}
\usepackage{mathtools}
\usepackage{amsthm}
\usepackage{bm} 
\usepackage{bbm} 

\usepackage[capitalize,noabbrev]{cleveref}

\theoremstyle{plain}
\newtheorem{theorem}{Theorem}[section]

\theoremstyle{definition}

\theoremstyle{remark}

\usepackage[textsize=tiny]{todonotes}

\author{
Zhengyang Zhou\textsuperscript{a}, 
Yunrui Li\textsuperscript{a}, 
Pengyu Hong\textsuperscript{a}, 
Hao Xu\textsuperscript{b}\thanks{Corresponding Author: \texttt{haxu@bwh.harvard.edu}} \\
\textsuperscript{a}Department of Computer Science, Brandeis University, Waltham, MA 02453, USA \\
\textsuperscript{b}Department of Medicine, Harvard Medical School, Boston, MA 02115, USA
}
\title{Multimodal Fusion with Relational Learning for Molecular Property Prediction }

\begin{document}

\maketitle







\begin{abstract}
Graph-based molecular representation learning is essential for predicting molecular properties in drug discovery and materials science. Despite its importance, current approaches struggle with capturing the intricate molecular relationships and often rely on limited chemical knowledge during training. Multimodal fusion, which integrates information from molecular graph and other data modalities, has emerged as a promising avenue for enhancing molecular property prediction. However, existing studies have explored only a narrow range of modalities, and the optimal integration stages for multimodal fusion remain largely unexplored. Furthermore, the reliance on auxiliary modalities poses challenges, as such data is often unavailable in downstream tasks. Here, we present MMFRL (Multimodal Fusion with Relational Learning), a framework designed to address these limitations by leveraging relational learning to enrich embedding initialization during multimodal pre-training. MMFRL enables downstream models to benefit from auxiliary modalities, even when these are absent during inference. We also systematically investigate modality fusion at early, intermediate, and late stages, elucidating their unique advantages and trade-offs. Using the MoleculeNet benchmarks, we demonstrate that MMFRL significantly outperforms existing methods with superior accuracy and robustness. Beyond predictive performance, MMFRL enhances explainability, offering valuable insights into chemical properties and highlighting its potential to transform real-world applications in drug discovery and materials science.

\end{abstract}

\section{Introduction}

\label{submission}

Graph representation learning for molecules  has gained significant attention in drug discovery and materials science, as it effectively encapsulates molecular structures and enables the effective investigation of structure-activity relationships~\citep{schneider2020rethinking,wieder2020compact, zhang2022graph, fang2022geometry, wang2023motif,chen2024drugdagt}. In this paradigm, atoms are usually treated as nodes and chemical bonds as edges, effectively encapsulating the connectivity that define molecular behaviors. However, it poses significant challenges due to intricate relationships among molecules and the limited chemical knowledge utilized during training.


\begin{figure*}[ht!]
    \centering
    \includegraphics[width=0.95\textwidth]{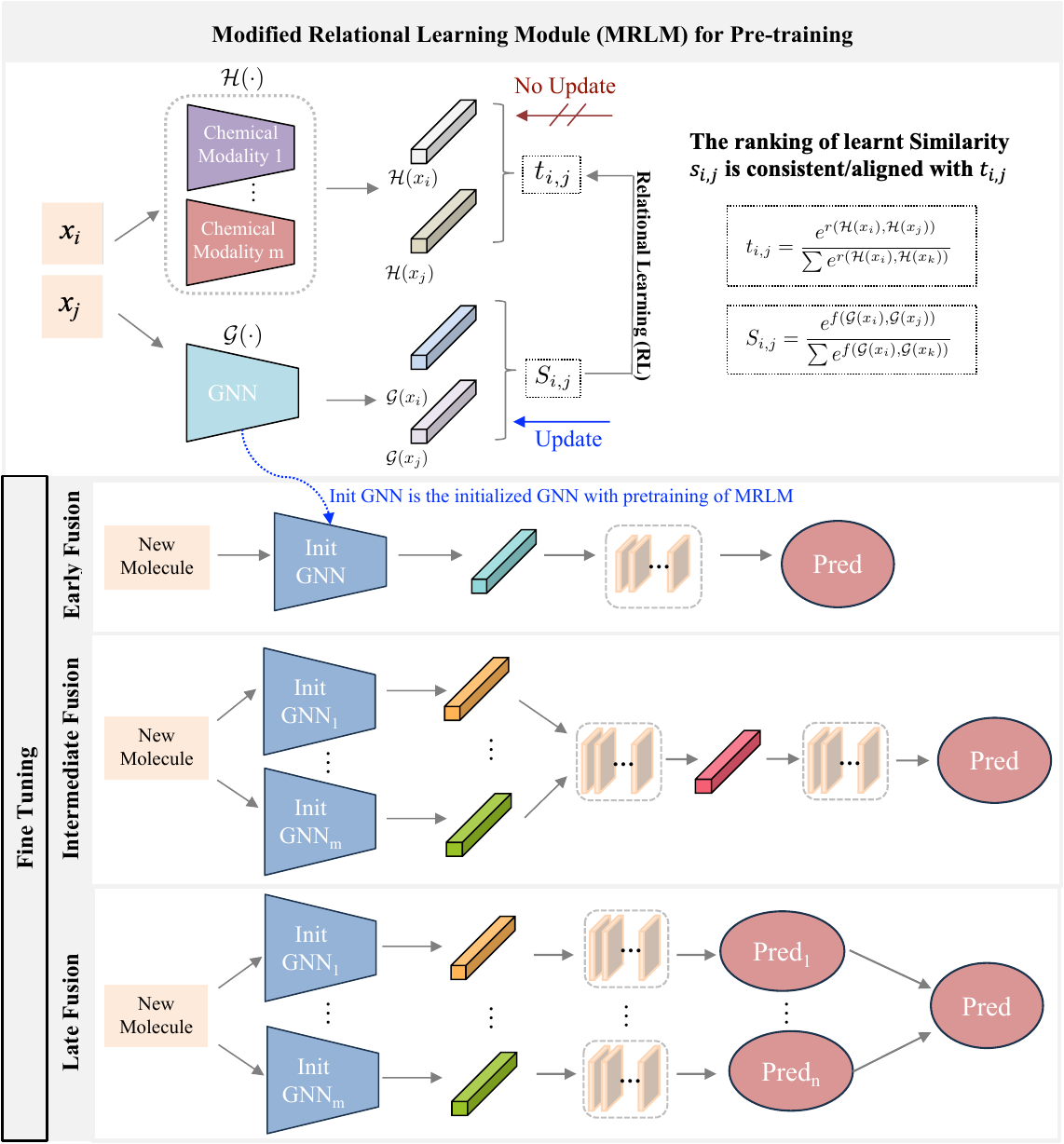}
    \caption{\textbf{Multimodal Fusion with Relational Learning for Molecular Property Prediction (MMFRL).} This figure shows our proposed idea about how to transfer the knowledge from other modalities and use fusion to improve the performance further. Unlike the general contrastive learning framework shown in Appendix Figure \ref{fig:traditional-cl}, MMFRL does not need to define positive or negative pairs and is capable of learning continuous ordering from target similarity. In Early Fusion, a single Init GNN is created by combining all modality information during pretraining. In Intermediate and Late Fusion, each modality has its own initialized GNN.}
    \label{fig:main-structure}
\end{figure*}

Contrastive Learning (CL) is often employed to study relationships among molecules. The primary focus within the domain of contrastive learning applied to molecular graphs centers on 2D-2D graph comparisons. Noteworthy representative examples: InfoGraph ~\citep{sun2019infograph} maximizes the mutual information between the representations of a graph and its substructures to guide the molecular representation learning;
GraphCL ~\citep{you2020graph}, MoCL ~\citep{sun2021mocl}, and MolCLR ~\citep{wang2022molecular} employs graph augmentation techniques to construct positive pairs; MoLR ~\citep{wang2022chemicalreactionaware}  establishes positive pairs with reactant-product relationships. In addition to 2D-2D graph contrastive learning, there are also noteworthy efforts exploring 2D-3D and 3D-3D contrastive learning in the field.
3DGCL ~\citep{moon20233d} is 3D-3D contrastive learning model, establishing positive pairs with conformers from the same molecules. GraphMVP ~\citep{liu2022pretraining}, GeomGCL ~\citep{li2022geomgcl}, and 3D Informax ~\citep{stark20223d} proposes 2D–3D view contrastive learning approaches. To conclude, 2D-2D and 3D-3D comparisons are intra-modality contrastive learning, as only one graph encoder is employed in these studies. However, these approaches often focus on  the motif and graph levels, leaving atom-level contrastive learning less explored. For example, consider Thalidomide: while the \textit{(R)}- and \textit{(S)}-enantiomers share the same topological graph and differ only at a single chiral center, their biological activities are drastically different—the \textit{(R)}-enantiomer is effective in treating morning sickness, whereas the \textit{(S)}-enantiomer causes severe birth defects. In other words, the \textit{(R)}- and \textit{(S)}-enantiomers are similar in terms of topological stucture but dissimilar in terms of  biological activities. Thus, a more sophisticated approach is required to tackle these scenarios. A potential solution would be to use continuous metrics within a multi-view space, enabling a more comprehensive understanding of these complex molecular relationships.

There are multiple approaches for such similarity learning. One approach of them is instance-wise discrimination, which involves directly assessing the similarity between instances based on their latent representations or features. ~\citep{wu2018unsupervised}. Naive instance-wise discrimination relies on pairwise similarity, leading to the development of contrastive loss ~\citep{hadsell2006dimensionality}. Although there are improved loss functions such as triplet loss ~\citep{hoffer2015deep}, quadruplet loss ~\citep{law2013quadruplet}, lifted structure loss ~\citep{oh2016deep}, N-pairs loss ~\citep{sohn2016improved}, and angular loss ~\citep{wang2017deep}, these methods still fall short in thoroughly capturing relationships among multiple instances simultaneously ~\citep{wang2019multi}. To address this limitation, a joint multi-similarity loss has been proposed, incorporating pair weighting for each pair to enhance instance-wise discrimination ~\citep{wang2019multi, zhang2021jointly}. However, these pair weightings requires the manual categorization of negative and positive pairs, as distinct weights are assigned to losses based on their categories. In this case, we can borrow the idea of Relational Learning ~\citep{zheng2021ressl} from computer vision by using different augmented views of the same instance tasks to similar features, while allowing for some variability. This approach captures the essential characteristics of the instance in a continuous scale, promoting relative consistency across the views without requiring them to be identical. By doing so, it enhances the model's ability to generalize and recognize underlying patterns in the data.

Besides, in order to enable a multi-view analysis from diverse sources is essential for improving molecule analysis, we can apply the Multi-Modality Fusion~\citep{lahat2015multimodal, khaleghi2013multisensor, poria2015deep, ramachandram2017deep, Pawlowski2023, manzoor2023multimodality, priessner2024enhancing}. It combines diverse heterogeneous data (e.g. text, images, graph) to create a more comprehensive understanding of complex scenarios . This approach leverages the strengths of each modality, potentially improving performance in tasks like sentiment analysis or medical diagnosis. Although challenging to implement due to the need to align different data streams, successful fusion can provide insights that that surpass those obtainable from individual modalities, advancing AI and data-driven decision-making. In particular, the way to fuse different modality should also depends on the dominance of each modality ~\citep{Pawlowski2023}. However, when it comes to multimodal learning for molecules, we often encounter data availability and incompleteness issues. This raises a critical question: how can multimodal information be effectively leveraged for molecular property reasoning when such data is absent in downstream tasks? Recent studies have demonstrated the effectiveness of pretraining molecular Graph Neural Networks (GNNs) by integrating additional knowledge sources~\citep{wang2021molecular, wang2022molecular, liu2022attention, xu2023asymmetric}. Building on this foundation, a promising solution is to pretrain multiple replicas of molecular GNNs, with each replica dedicated to learning from a specific modality. This approach allows downstream tasks to benefit from multimodal data that is not  accessible during fine-tuning, ultimately improving representation learning.

Facing these challenges and opportunities, we propose  MMFRL  (Multimodal Fusion
with Relational Learning for Molecular Property Prediction), a novel framework
features relational learning (RL) and multimodal fusion (MMF). RL utilizes a continuous relation metric to evaluate relationships among instances in the feature space  ~\citep{balcan2006theory, wen2023pairwise}. Our major contribution comprises three aspects: 
\textbf{\textit{Conceptually:}} We introduce a modified relational learning metric for molecular graph representation that offers a more comprehensive and continuous perspective on inter-instance relations, effectively capturing both localized and global relationships among instances. To the best of our knowledge, this is the first work to demonstrate such generalized relational learning metric for molecular graph representation.
\textbf{\textit{Methodologically:}} 
Our proposed modified relational metric captures complex relationships by converting pairwise self-similarity into relative similarity, which evaluates how the similarity between two elements compares to the similarity of other pairs in the dataset. In addition, we integrate these metrics into a fused multimodal representation, which has the potential to enhance performance, allowing downstream tasks to leverage modalities that are not directly accessible during fine-tuning.
\textbf{\textit{Empirically:}} 
MMFRL excels in various downstream tasks for Molecular Property Predictions. Last but not least, we demonstrate the explainability of the learned representations through two post-hoc analysis. Notably, we explore minimum positive subgraphs and maximum common subgraphs to gain insights for further drug molecule design.

\section{Results}

\subsection{The effectiveness of pre-training}
We first illustrate the impact of pre-training initialization on performance on DMPNN ~\citep{yang2019analyzing}. As shown in Table \ref{tab:multimodal-metrics-single-level}, the average performance of pre-trained models outperform the non-pre-trained model in all tasks except for Clintox. The results of various downstream tasks indicate that different tasks may prefer different modalities. Notably, the model pre-trained with the NMR modality achieves the highest performance across three classification tasks. Similarly, the model pre-trained with the Image modality excels in three tasks, two of which are regression tasks related to solubility, aligning with findings from prior literature \citep{xu2023asymmetric}. Additionally, the model pre-trained with The fingerprint method achieves the best performance in two tasks, including MUV, which has the largest dataset. 

\setlength{\tabcolsep}{1.5pt}
\renewcommand{\arraystretch}{1.1}
\begin{table*}[t]
\caption{Study on the performances of $\text{MMFRL}_{Unimodality}$. The best results are denoted in bold, and the second-best are indicated with underlining among the five modalities. The first 8 tasks are for classification under evaluation of ROC-AUC, while the last three are for regression with evaluation of RMSE. }
\label{tab:multimodal-metrics-single-level}
\begin{center}
\begin{tiny}
\begin{sc}
\begin{tabular}{lcccccccc|ccc}
\toprule
Data Set & BBBP & bace & Sider & Clintox & HIV & MUV & Tox21 & ToxCast & ESOL & FreeSolv & Lipo \\
\midrule
SMILES & {92.9$\pm$1.5} & 90.9$\pm$3.3 & 64.9$\pm$0.3 & 78.2$\pm$1.9  & \textbf{83.3$\pm$1.1} & 80.1$\pm$2.5 & \underline{85.7$\pm$1.2} &70.5$\pm$2.5 & 0.811$\pm$ 0.109 & \underline{1.623$\pm$ 0.168} & \underline{0.539$\pm$ 0.017}\\

$\text{NMR}_\text{spectrum}$ & 91.0$\pm$2.0 & \textbf{93.2$\pm$2.7} & \textbf{68.1$\pm$1.5} & \textbf{87.7$\pm$6.5}  & 80.9$\pm$5.0 & \underline{80.9$\pm$5.0} &  85.1$\pm$0.4 & \textbf{71.1$\pm$0.8} & 0.844$\pm$ 0.123 & 2.417$\pm$ 0.495 & 0.609$\pm$ 0.031 \\

Image & \underline{93.1$\pm$2.4} & \underline{92.9$\pm$1.8} & 65.3$\pm$1.5 & 86.2$\pm$6.5 & \underline{82.3$\pm$0.6} & 78.7$\pm$1.7 & \textbf{86.0$\pm$1.0} & \underline{71.0$\pm$1.6} & \textbf{0.761$\pm$ 0.068} & 1.648$\pm$ 0.045 & \textbf{0.537$\pm$ 0.005} \\

Fingerprint & 92.9$\pm$2.3 & 91.7$\pm$3.6  & \underline{65.6$\pm$0.7} & \underline{87.5$\pm$6.0} & 81.2$\pm$2.5 & \textbf{82.9$\pm$3.1} & 85.3$\pm$1.3 &70.0$\pm$1.4 & \underline{0.808$\pm$ 0.071} & \textbf{1.437$\pm$ 0.134} & 0.565$\pm$ 0.017\\

$\text{NMR}_\text{Peak}$ & \textbf{93.4$\pm$2.7} & 89.3$\pm$1.7 & 62.8$\pm$2.1 & 86.1$\pm$5.4& 82.1$\pm$0.4 &75.4$\pm$5.2 &84.9$\pm$1.0 & 70.6$\pm$0.8 & 0.924$\pm$0.083 &1.707$\pm$0.126 & 0.587$\pm$0.021\\
\hline

Average & 92.8$\pm$1.9 & 91.4$\pm$2.7 & 65.3$\pm$2.0 & 85.0$\pm$5.7 & 81.8$\pm$2.2 & 79.4$\pm$4.0 & 85.4$\pm$0.9 & 70.6$\pm$1.3 & 0.830$\pm$0.094 & 1.766$\pm$0.394 & 0.586$\pm$0.048\\

\textcolor{black}{Max} & \textcolor{black}{93.4$\pm$2.7} & \textcolor{black}{93.2$\pm$2.7} & \textcolor{black}{68.1$\pm$1.5} & \textcolor{black}{87.7$\pm$6.5} & \textcolor{black}{83.3$\pm$1.1} & \textcolor{black}{82.9$\pm$3.1} & \textcolor{black}{86.0$\pm$1.0} & \textcolor{black}{71.1$\pm$0.8} & \textcolor{black}{0.761$\pm$ 0.068}  &  \textcolor{black}{1.437$\pm$ 0.134} & \textcolor{black}{0.537$\pm$ 0.005} \\
No Pre-training & 91.9$\pm$3.0 & 85.2$\pm$0.6 & 57.0$\pm$0.7 & 90.6$\pm$0.6 & 77.1$\pm$0.5 & 78.6$\pm$1.4 & 75.9$\pm$0.7 & 63.7$\pm$0.2 & 1.050$\pm$0.008 & 2.082$\pm$0.082 & 0.683$\pm$0.016 \\ 

\bottomrule
\end{tabular}
\end{sc}
\end{tiny}
\end{center}
\end{table*}

\subsection{Overall performance of $\text{MMFRL}$}
As shown in Table~\ref{table:overall_performance} and Table~\ref{table:overall_performance_regression}, MMFRL demonstrates superior performance compared to all baseline models and the average performance of DMPNN pretrained with extra modalities across all 11 tasks evaluated in MoleculeNet. {Results in Tables~\ref{tab:dude_auc_comparison} and Table~\ref{tab:performance_lit_pcba_transposed} demonstrates our great performance compared to the baseline models on the Dud-E \citep{mysinger2012} and LIT-PCBA \citep{trannguyen2020} datasets.} This robust performance highlights the effectiveness of our approach in leveraging multimodal data. In particular, while individual models pre-trained on other modalities for ClinTox fail to outperform the No-pretraining model, the fusion of these pre-trained models leads to improved performance. Besides, apart from Tox21 and Sider, the fusion models significantly enhances overall performance. In particular, the intermediate fusion model stands out by achieving the highest scores in seven distinct tasks, showcasing its ability to effectively combine features at a mid-level abstraction. the late fusion model achieves the top performance in two tasks. These results underscore the advantages of utilizing various fusion strategies in multimodal learning, further validating the efficacy of the MMFRL framework.

\setlength{\tabcolsep}{5pt}
\begin{table*}[t!]
\caption{Overall performances (ROC-AUC) on classification downstream tasks. The best results are denoted in bold, and the second-best are indicated with underlining. For early fusion of MMFRL, all the predefined weight of each modality are 0.2. (Note: N-Gram is highly time-consuming on ToxCast.) }
\vspace{2pt}
\label{table:overall_performance}
\begin{center}
\begin{scriptsize}
\begin{sc}
\begin{tabular}{lcccccccc}
\toprule
Data Set & BBBP & bace & Sider & Clintox & HIV & MUV & Tox21 & ToxCast \\
\midrule
AttentiveFP & 64.3$\pm$1.8 & 78.4$\pm$2.2  & 60.6$\pm$3.2 & 84.7$\pm$0.3 & 75.7$\pm$1.4 & 76.6$\pm$1.5 & 76.1$\pm$0.5 & 63.7$\pm$0.2  \\


DMPNN & 91.9$\pm$3.0 & 85.2$\pm$0.6 & 57.0$\pm$0.7 & 90.6$\pm$0.6 & 77.1$\pm$0.5 & 78.6$\pm$1.4 & 75.9$\pm$0.7 & 63.7$\pm$0.2  \\


\text{N-Gram} & 91.2$\pm$0.3 & 79.1$\pm$1.3 & 63.2$\pm$0.5 & 87.5$\pm$2.7 & 78.7$\pm$0.4  & 76.9$\pm$0.7 & 76.9$\pm$2.7 & -\\


GEM & 72.4$\pm$0.4 & 85.6$\pm$1.1 & \underline{67.2$\pm$0.4} & 90.1$\pm$1.3 & 80.6$\pm$0.9 & 81.7$\pm$0.5 & 78.1$\pm$0.1 & 69.2$\pm$0.4\\

Uni-Mol & 72.9$\pm$0.6 & 85.7$\pm$0.2 & 65.9$\pm$1.3 & \underline{91.9$\pm$1.8} & 80.8$\pm$0.3 & {82.1$\pm$1.3} & 79.6$\pm$0.5 & 69.6$\pm$0.1 \\


InfoGraph & 69.2$\pm$0.8 & 73.9$\pm$2.5 & 59.2$\pm$0.2 & 75.1$\pm$5.0 & 74.5$\pm$1.8 &  74.0$\pm$1.5 & 73.0$\pm$0.7 & 62.0$\pm$0.3 \\

GraphCL & 67.5$\pm$3.3 & 68.7$\pm$7.8 & 60.1$\pm$1.3 & 78.9$\pm$4.2 & 75.0$\pm$0.4 & 77.1$\pm$1.0 & 75.0$\pm$0.3 & 62.8$\pm$0.2 \\


MolCLR & 73.3$\pm$1.0 & 82.8$\pm$0.7 & 61.2$\pm$3.6 & 89.8$\pm$2.7 & 77.4$\pm$0.6 & 78.9$\pm$2.3 & 74.1$\pm$5.3 & 65.9$\pm$2.1 \\

$\text{MolCLR}_{\text{cmpnn}}$ & 72.4$\pm$0.7 & 85.0$\pm$2.4 & 59.7$\pm$3.4 & 88.0$\pm$4.0 & 77.8$\pm$5.5 & 74.5$\pm$2.1 & 78.4$\pm$2.6 & 69.1$\pm$1.2 \\

GraphMVP & 72.4$\pm$1.6 & 81.2$\pm$9.0 & 63.9$\pm$1.2 & 79.1$\pm$2.8 & 77.0$\pm$1.2 & 77.7$\pm$6.0 & 75.9$\pm$5.0 & 63.1$\pm$0.4\\

\hline

$\text{Unimodality}_{avg}$& 92.8$\pm$1.9 & 91.4$\pm$2.7 & {65.3$\pm$2.0} & 85.0$\pm$5.7 & 81.8$\pm$2.2 & 79.4$\pm$4.0 & \underline{85.4$\pm$0.9} & {70.6$\pm$1.3} \\

$\textcolor{black}{Unimodality}_{\textcolor{black}{Max}}$& \textcolor{black}{93.4$\pm$2.7} & \textcolor{black}{93.2$\pm$2.7} & \textbf{\textcolor{black}{68.1$\pm$1.5}} & \textcolor{black}{87.7$\pm$6.5} & \textbf{\textcolor{black}{83.3$\pm$1.1}} & \underline{\textcolor{black}{82.9$\pm$3.1}} & \textbf{\textcolor{black}{86.0$\pm$1.0}} & \underline{\textcolor{black}{71.1$\pm$0.8}} \\

$\text{MMFRL}_{early}$& 91.6$\pm$5.0 & \underline{94.3$\pm$2.4} & {66.4$\pm$1.9} & 85.3$\pm$6.8 & {82.0$\pm$2.4} & 80.6$\pm$3.2 & {85.2$\pm$0.2} & 69.8$\pm$1.1 \\

$\text{MMFRL}_{intermediate}$& \textbf{95.4$\pm$0.7} & \textbf{95.1$\pm$1.0} & {64.3$\pm$1.2} & \textbf{93.4$\pm$1.1} & {81.2$\pm$1.3} &\textbf{83.5$\pm$1.6} & {85.1$\pm$0.1} & \textbf{71.9$\pm$1.1}  \\


$\text{MMFRL}_{late}$& \underline{94.7$\pm$0.6} & {91.6$\pm$2.6} & {64.2$\pm$1.2} & {87.0$\pm$0.4} & \underline{82.9$\pm$0.2} &82.1$\pm1.7$ & {77.7$\pm$0.5} & {70.2$\pm$0.3}  \\

\bottomrule
\end{tabular}
\end{sc}
\end{scriptsize}
\end{center}
\end{table*}
\vspace{-10pt} 
\begin{table}[t]
\setlength{\tabcolsep}{6pt}
\caption{Overall performances (RMSE) on regression downstream tasks. The best results are denoted in bold, and the second-best are indicated with underlining. For early fusion of MMFRL, all the predefined weight of each modality are 0.2.}
\label{table:overall_performance_regression}
\begin{center}
\begin{scriptsize}
\begin{tabular}{lccc}
\toprule
Data Set & ESOL & FreeSolv & Lipo  \\
\midrule
AttentiveFP & 0.877$\pm$0.029 & 2.073$\pm$0.183 & 0.721$\pm$0.001 \\

DMPNN & 1.050$\pm$0.008 & 2.082$\pm$0.082 & 0.683$\pm$0.016\\

$\text{N-Gram}_{\text{RF}}$  & 1.074$\pm$0.107 & 2.688$\pm$0.085 & 0.812$\pm$0.028 \\

$\text{N-Gram}_{\text{XGB}}$ & 1.083$\pm$0.082 & 5.061$\pm$0.744 & 2.072$\pm$0.030   \\

GEM & 0.798$\pm$0.029 & 1.877$\pm$0.094 & 0.660$\pm$0.008 \\

Uni-Mol & {0.788$\pm$0.029} & 1.620$\pm$0.035 & \textcolor{black}{0.603$\pm$0.010} \\





MolCLR & 1.113$\pm$0.023 & 2.301$\pm$0.247 & 0.789$\pm$0.009  \\

$\text{MolCLR}_{\text{CMPNN}}$  & 0.911$\pm$0.082 & 2.021$\pm$0.133 & 0.875$\pm$0.003 \\


\hline
$\text{Unimodality}_{avg}$ & 0.924$\pm$0.083 & {1.707$\pm$0.126} & 0.587$\pm$0.021  \\
\textcolor{black}{${Unimodality}_{Max}$} & \underline{\textcolor{black}{0.761$\pm$ 0.068}}  &  \textbf{\textcolor{black}{1.437$\pm$ 0.134}} & \underline{\textcolor{black}{0.537$\pm$ 0.005}} \\
$\text{MMFRL}_{early}$ & 1.037$\pm$0.170 & 2.093$\pm$0.090 & 0.607$\pm$0.034 \\
$\text{MMFRL}_{intermediate}$ & \textbf{0.730$\pm$0.019} & \underline{1.465$\pm$0.096} & {0.552$\pm$0.014} \\
$\text{MMFRL}_{late}$ & {0.763$\pm$0.035} & 1.741$\pm$0.191 & \textbf{0.525$\pm$0.018} \\
\bottomrule
\end{tabular}
\vspace{-8pt} 
\end{scriptsize}
\end{center}
\end{table}

\setlength{\tabcolsep}{1.5pt}
\renewcommand{\arraystretch}{1.1}
\begin{table}[t]
\caption{{Study on the performances of $\text{MMFRL}_{Intermediate}$ with different contrastive loss functions: Contrastive Loss (CL) and Tripelet Loss (TL). The best results are denoted in bold. The first 8 datasets are classification tasks evaluated using ROC-AUC, while the last three are regression tasks evaluated using RMSE. Our model outperforms other loss functions in most of the datasets.}}
\begin{center}
\begin{tiny}
\begin{sc}
\begin{tabular}{lcccccccc|ccc}
\toprule
Data Set & BBBP & bace & Sider & Clintox & HIV & MUV & Tox21 & ToxCast & ESOL & FreeSolv & Lipo \\
\midrule
$\text{MMFRL}_{intermediate}$& \textbf{95.4$\pm$0.7} & \textbf{95.1$\pm$1.0} & \textbf{64.3$\pm$1.2} & \textbf{93.4$\pm$1.1} & \textbf{81.2$\pm$1.3} &\textbf{83.5$\pm$1.6} & {85.1$\pm$0.1} & \textbf{71.9$\pm$1.1} & \textbf{0.730$\pm$0.019} & \textbf{1.465$\pm$0.096} & \textbf{0.552$\pm$0.014}   \\
\textcolor{black}{$\text{MMFRL}_{intermediate CL}$}&\textcolor{black}{93.2$\pm$1.5} & \textcolor{black}{89.7$\pm$1.3} & \textcolor{black}{61.1$\pm$2.6} & \textcolor{black}{90.6$\pm$1.7} & \textcolor{black}{80.9$\pm$1.7} & \textcolor{black}{78.2$\pm$1.3} & \textbf{\textcolor{black}{85.7$\pm$0.4}} & \textcolor{black}{70.8$\pm$0.8} & \textcolor{black}{0.792$\pm$0.034} & \textcolor{black}{2.094$\pm$0.377} & \textcolor{black}{0.609$\pm$0.022}  \\

\textcolor{black}{$\text{MMFRL}_{intermediate TL}$}& \textcolor{black}{93.2$\pm$2.2} & \textcolor{black}{91.3$\pm$1.2} & \textcolor{black}{61.8$\pm$1.7} & \textcolor{black}{91.8$\pm$2.5} & \textcolor{black}{80.0$\pm$1.5} & \textcolor{black}{78.8$\pm$0.2} & \textcolor{black}{85.6$\pm$0.5} & \textcolor{black}{70.7$\pm$0.4} &  \textcolor{black}{0.780$\pm$0.037} & \textcolor{black}{2.072$\pm$0.199} & \textcolor{black}{0.577$\pm$0.005} \\
\bottomrule
\end{tabular}
\end{sc}
\end{tiny}
\end{center}
\label{tab:intermediate-contrastive-loss-ablation}
\end{table}

\begin{table}[ht]
\centering

\caption{{Comparison of Average AUC-ROC across DUD-E dataset}}
{
\begin{tabular}{l c}
\toprule
\textbf{Method} & \textbf{Average AUC-ROC} \\
\midrule
COSP \citep{gao2022cosp} & 90.10 \\
Graph CNN \citep{torng2019gcn} & 88.60 \\
Drug VQA \citep{zheng2020vqas} & 97.20\\
AttentionSiteDTI \citep{yazdani2022attentionsite} & \underline{97.10} \\
$\text{DrugCLIP}_{FT}$ \citep{gao2023drugclip} & 96.59 \\
{MMRFL-Intermediate (Ours)} & \textbf{98.32} \\
{MMRFL-Late (Ours)} & {96.78} \\
\bottomrule
\end{tabular}
}
\label{tab:dude_auc_comparison}
\end{table}

\begin{table*}[ht]
\centering

\caption{{Overall performances (ROC-AUC) on classification downstream tasks for $\text{Lit-PCBA}$. The performance that are not by us is from the paper \cite{cai2022fpgnn}}}
{
\begin{tabular}{lcccccccccc}
\toprule
Task & NB & SVM & RF & XGBoost & DNN & GCN & GAT & FP-GNN & MMRFL\_Int & MMRFL\_lat \\
\midrule
ADRBZ & 55.2 & 53.4 & 49.8 & 50.0 & 83.3 & \underline{83.7} & 76.8 & \textbf{88.6} & 76.3 & 76.0 \\
ALDH1 & 69.3 & 76.0 & 74.1 & 75.0 & 75.6 & 73.0 & 73.9 & 76.6 & \textbf{78.8} & \underline{78.7} \\
ESR1\_ago & 66.1 & 55.2 & 44.8 & 50.0 & 69.0 & 58.7 & 71.3 & \underline{72.8} & \textbf{76.8} & 31.5 \\
ESR1\_ant & 54.3 & 63.0 & 53.3 & 52.8 & 58.2 & \textbf{67.1} & \underline{65.6} & 64.2 & 54.3 & 58.8 \\
FEN1 & 87.6 & 87.7 & 65.7 & 88.8 & \underline{90.1} & 89.7 & 88.8 & 88.9 & \textbf{90.5} & 85.4 \\
GBA & 70.9 & \underline{77.8} & 59.9 & \textbf{83.0} & 77.7 & 73.5 & 77.6 & 75.1 & 77.3 & 75.3 \\
IDH1 & \textbf{88.7} & 80.7 & 49.8 & 50.0 & 67.8 & 81.3 & \underline{86.1} & 78.7 & 71.0 & 41.6 \\
KAT2A & 65.9 & 61.2 & 53.7 & 50.0 & 59.5 & 62.1 & 66.2 & 63.2 & \textbf{71.6} & \underline{68.7} \\
MAPK1 & 68.6 & 66.5 & 57.9 & 59.3 & 70.8 & 66.8 & 69.7 & \textbf{77.1} & \underline{73.0} & 69.6 \\
MTORC1 & 59.8 & 59.1 & 53.2 & 63.9 & 63.4 & \textbf{66.9} & 61.5 & 58.3 & \underline{62.0} & 59.0 \\
OPRK1 & 53.8 & 53.2 & 49.8 & 50.0 & \underline{71.0} & 64.4 & 63.6 & 54.5 & \textbf{72.8} & 68.6 \\
PKM2 & 68.4 & \underline{75.3} & 58.1 & 73.7 & 71.9 & 63.6 & 72.4 & 73.2 & \textbf{77.2} & 71.4 \\
PPARG & 67.5 & 79.3 & 66.9 & 49.0 & \underline{81.1} & 78.8 & 78.4 & \textbf{82.9} & 69.7 & 69.5 \\
TP53 & 64.8 & 60.2 & 60.2 & 64.3 & 70.6 & 74.9 & 70.6 & \underline{76.3} & \textbf{80.7} & 57.3 \\
VDR & 80.4 & 69.0 & 64.4 & \underline{78.2} & \textbf{79.4} & 77.3 & 78.0 & 77.4 & 73.9 & 73.5 \\
\bottomrule
\end{tabular}
}
\label{tab:performance_lit_pcba_transposed}
\end{table*}

\setlength{\tabcolsep}{4pt}

\vspace{10pt}
\subsection{Analysis of the fusion effect}

\subsubsection{General Comparison among various Ways of Fusions}
Early Fusion is employed during the pretraining phase and is easy to implement, as it aggregates information from different modalities directly. However, its primary limitation lies in the necessity for predefined weights assigned to each modality. These weights may not accurately reflect the relevance of each modality for the specific downstream tasks, potentially leading to suboptimal performance.

Intermediate Fusion is able to capture the interaction between modalities early in the fine-tuning process, allowing for a more dynamic integration of information. This method can be particularly beneficial when different modalities provide complementary information that enhances overall performance. If the modalities effectively compensate for one another's strengths and weaknesses, Intermediate Fusion may emerge as the most effective approach.

In contrast, Late Fusion enables each modality to be explored independently, maximizing the potential of individual modalities without interference from others. This separation allows for a thorough examination of each modality's contribution. When certain modalities dominate the performance metrics, Late Fusion can capitalize on their strengths by effectively leveraging the most informative signals.. This approach is especially useful in scenarios where the dominance of specific modalities can be leveraged to enhance overall model performance. 

\textcolor{black}{In addition, we conduct an ablation study to evaluate the performance of our proposed loss functions against two traditional contrastive learning losses—Contrastive Loss and Triplet Loss—in the context of intermediate fusion. The experimental results as shown in Table~\ref{tab:intermediate-contrastive-loss-ablation} demonstrate that our proposed methods outperform the baseline approaches across the majority of tasks in the MoleculeNet dataset, thereby highlighting the superiority of our approach.}

\subsubsection{Explainability of Learnt Representations}
To demonstrate the interpretability of the representations learned by the proposed fusion strategies, we present the post-hoc analysis results on two tasks, ESOL and Lipo, as case studies. The results showcase that the learnt representations can capture task-specific patterns and offer valuable insights for molecular design.


\textbf{ESOL with Intermediate Fusion.} As presented in Table \ref{table:overall_performance_regression}, the intermediate fusion method ~\ref{sec:intermediate-fusion} exhibits superior performance on the ESOL regression task for predicting solubility. To further analyze this performance, we employed t-SNE to reduce the dimensionality of the molecule embeddings from 300 to 2, resulting in a heatmap visualized in Figure \ref{fig:ESOL-visualization}. The embeddings derived from individual modalities prior to fusion do not display a clear pattern, showing no smooth transition from low to high solubility. In contrast, the embeddings by intermediate fusion reveal a distinct and smooth transition in solubility values: molecules with similar solubility cluster together, forming a gradient that extends from the bottom left (indicating lower solubility) to the upper center (representing higher solubility). This trend underscores the effectiveness of the intermediate fusion approach in accurately capturing the quantitative structure-activity relationships for aqueous solubility.

Additionally, we examined the similarity between the respective embeddings prior to intermediate fusion and the resulting fused embedding, as depicted in Figure \ref{fig:ESOL-intermediate-fusion-similarity}. Our analysis indicates that the embeddings from each modality exhibit low similarity with the intermediate-fused representation. This observation suggests that the modalities complement each other, collectively enhancing the resulting representation of the intermediate-fused embedding.

\textbf{Lipo with Late Fusion. }As detailed in Table~\ref{table:overall_performance_regression}, the Late Fusion method (described in Section ~\ref{sec:late-fusion}) demonstrates superior performance on the Lipo regression task for predicting solubility in fats, oils, lipids, and non-polar solvents. According to Equation ~\ref{equ:late-fusion-formula}, the final prediction is determined by the respective coefficients ($w_i$) and predictions ($p_i$) from each modality.

Figure~\ref{fig:lipo-late-fusionvisualization} shows the distributions of the coefficient values, predictions, and their products for each modality. Notably, the SMILES and Image modalities display a wide range of values, highlighting their potential to significantly influence the final predictions. This observation aligns with the strong performance achieved when pretraining using either of these two modalities, as shown in Table~\ref{tab:multimodal-metrics-single-level}. In contrast, the $\text{NMR}_\text{Peak}$ values display a narrower range, indicating its role as a modifier for finer adjustments in the predictions. Furthermore, we observe that the contributions from $\text{NMR}_\text{Spectrum}$ and Fingerprint modalities are minimal, with their corresponding values approaching zero. This outcome highlights the advantages of the Late Fusion approach in effectively identifying and leveraging dominant modalities, thereby optimizing the overall predictive performance.

\begin{figure*}[ht]
    \centering
    \includegraphics[width=0.95\textwidth]{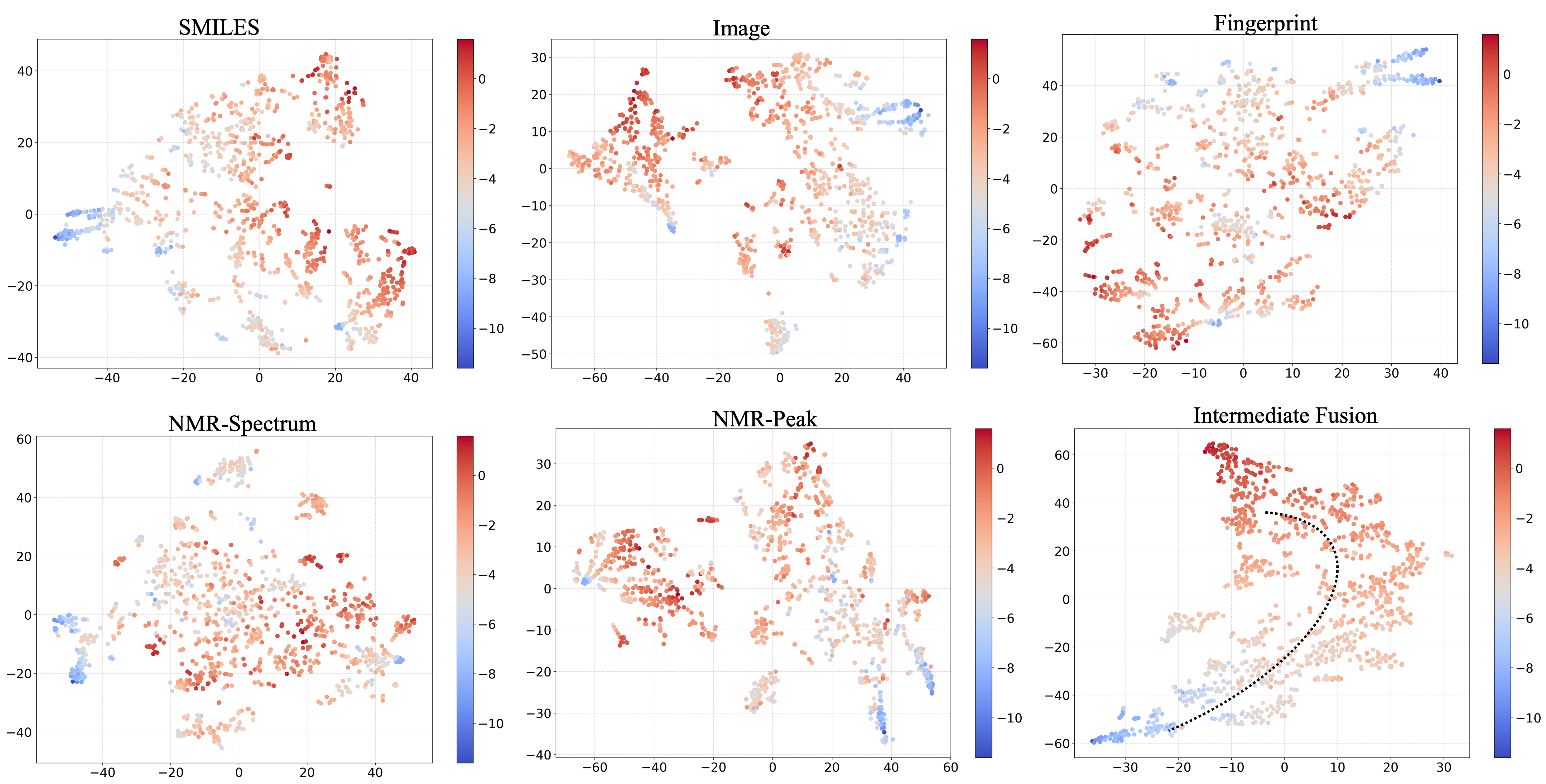}
    \caption{T-SNE visualization depicting the ESOL molecule embeddings for intermediate fusion in Section ~\ref{sec:intermediate-fusion} alongside molecules within the highlighted region. Each point in the heatmap corresponds to the embeddings of respective molecules in ESOL, with color indicating solubility levels. Red denotes higher solubility, while blue indicates lower solubility. The embeddings derived from individual modalities prior to fusion do not display a clear pattern, the embeddings by intermediate fusion forms a gradient that extends from the bottom left (indicating lower solubility) to the upper center (representing higher solubility).}
    \label{fig:ESOL-visualization}
\end{figure*}
\vspace{-5pt}
\begin{figure*}[ht]
    \centering
    \includegraphics[width=0.95\textwidth]{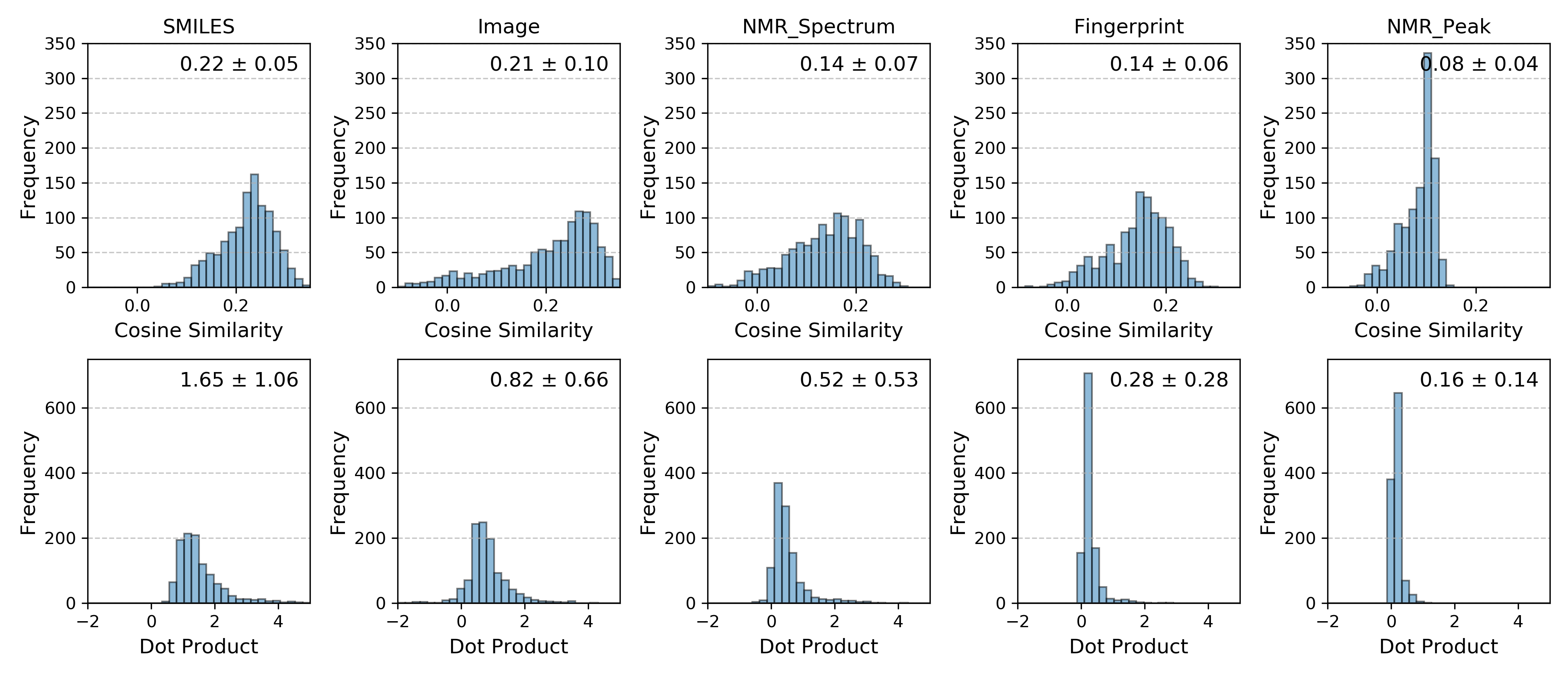}
    \caption{This figure shows the distribution of similarities between each modality and the intermediate fusion embedding for ESOL. In both Cosine Similarity and Dot Product, the embeddings from each modality exhibit low similarity with the intermediate-fused representation.}
    \label{fig:ESOL-intermediate-fusion-similarity}
\end{figure*}
\begin{figure}[t]
    \centering
    \includegraphics[width=0.95\textwidth]{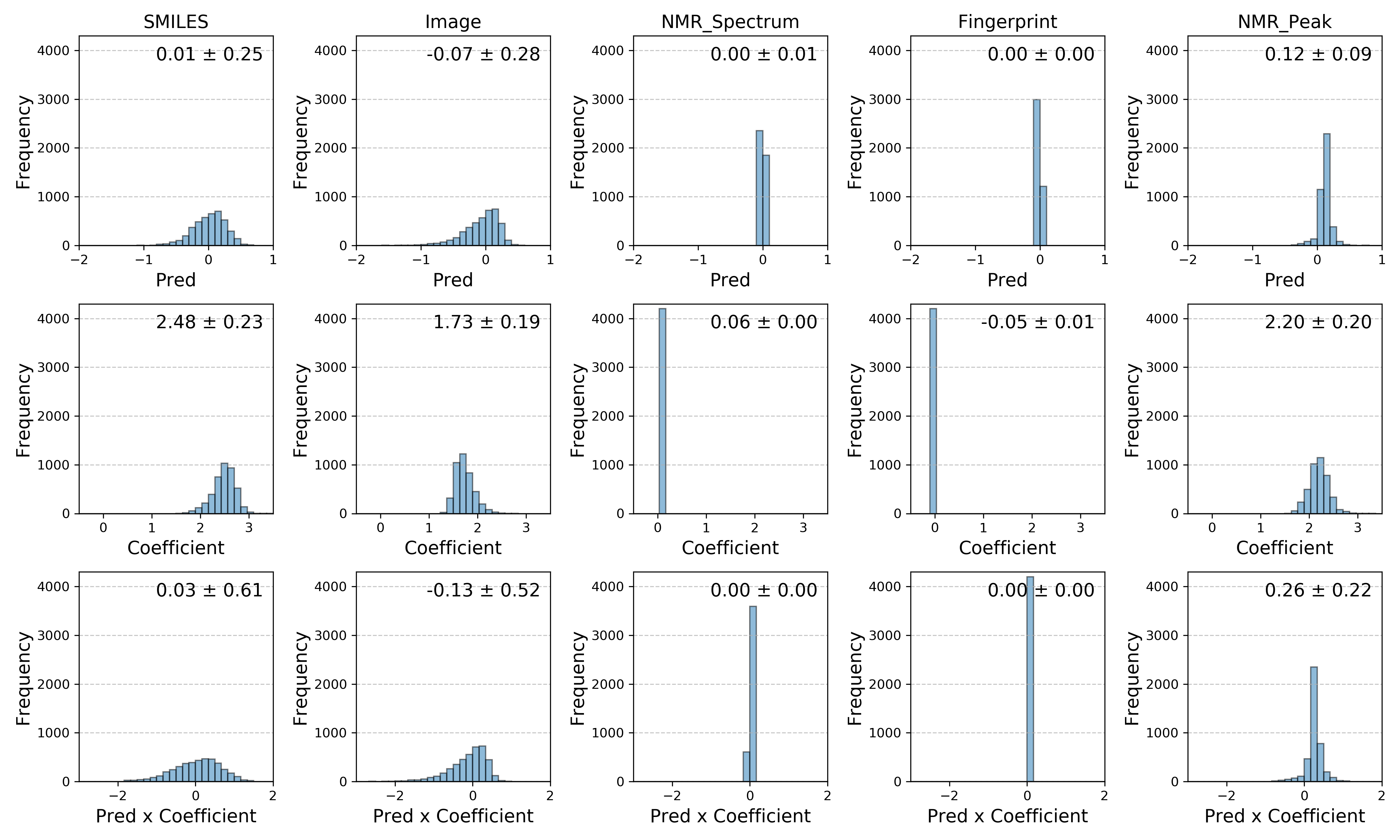}
    \caption{Lipo late fusion contribution analysis reveals that the three primary contributors are SMILES, image, and $\text{NMR}_\text{peak}$. In contrast, $\text{NMR}_\text{spectrum}$ and fingerprint exhibit negligible contributions.}
    \label{fig:lipo-late-fusionvisualization}
\end{figure}

\textcolor{black}{\textbf{\\Substructure analysis with BACE}. We explore the binding potential of positive inhibitor molecules targeting BACE and their associated key functional substructures, referred to as minimum positive subgraphs (MPS). To identify MPS, we employ a Monte Carlo Tree Search (MCTS) approach integrated into our BACE classification model, as implemented in RationalRL ~\citep{jin2020multiobjective}. MCTS, being an iterative process, allows us to evaluate each candidate substructure for its binding potential with our model. Following the determination of MPSs, we categorize the original positive BACE molecules based on their respective MPSs. By computing the binding potential difference between the original molecule and its MPS, we can identify structural features that contribute to changes in binding affinity as shown in Figure ~\ref{fig:bace-visualization}.}

\textcolor{black}{In the case of the MPS 5 group, the binding score is heavily influenced by steric effects. The top three high-performing designs (5a–5c) all feature a flexible and compact alkylated pyrazole structure (colored green), which likely facilitates better accommodation within the binding pocket. In contrast, the three lowest-performing designs (5n–5p) incorporate a more rigid and bulky (trifluoromethoxy)benzene moiety (colored red), which may introduce steric hindrance and reduce binding efficiency. Additionally, the pyrazole ring contains two nitrogen atoms, offering more potential for hydrogen bonding interactions with the target protein, whereas the (trifluoromethoxy)benzene group has only one oxygen atom, limiting its capacity for such interactions. This comparison highlights the importance of both molecular flexibility and functional group composition in optimizing binding affinity.}

\begin{figure*}[ht]
    \centering
    \includegraphics[width=0.95\textwidth]{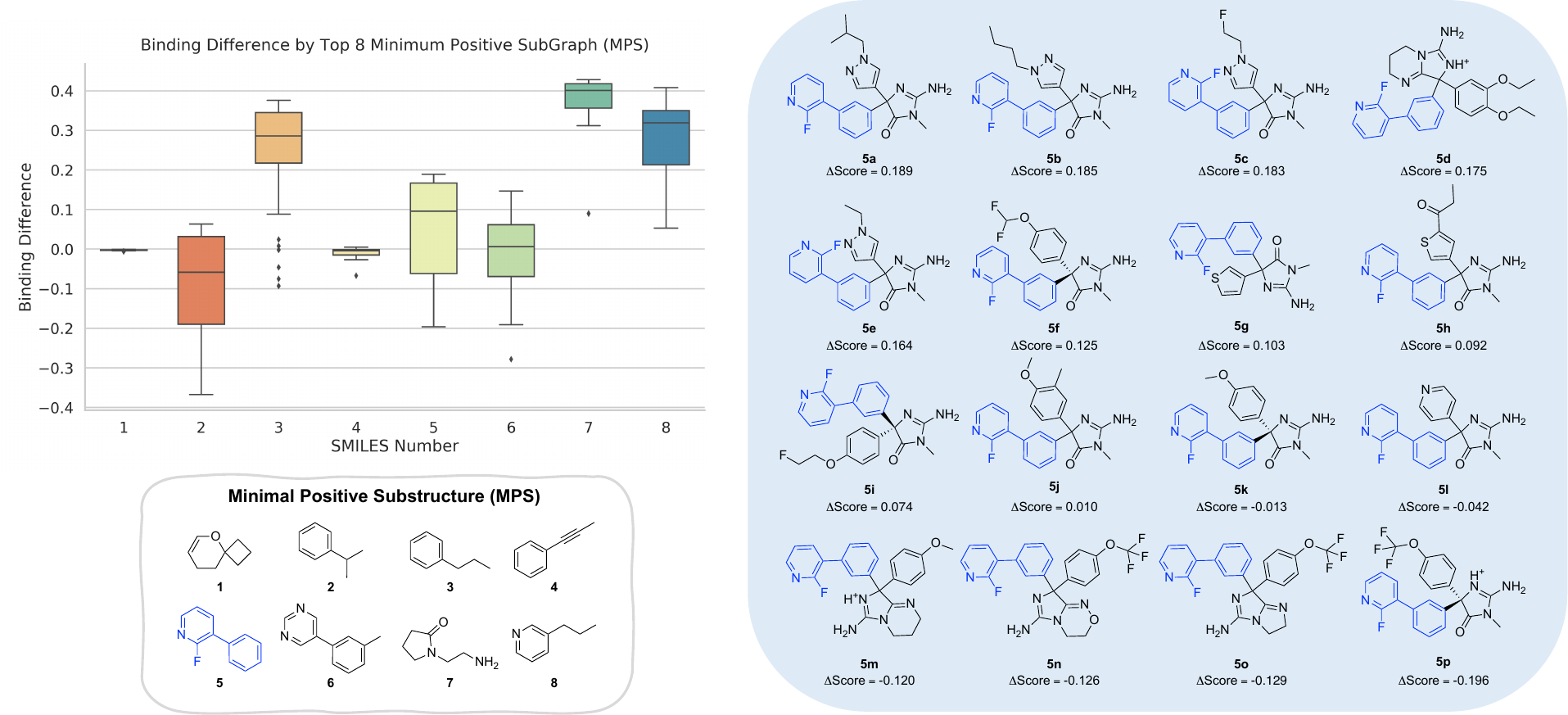}
    \caption{{The left sub-figure is the boxplot of the binding difference for the respective groups of molecules by the top 8 most frequent Minimum Positive Subgraph. The right sub-Figure showsthe detail strucutre of the 5th MPS.}}
    \label{fig:bace-visualization}
\end{figure*}

\textcolor{black}{
\subsubsection{Sensitivity Analysis}}
{Choosing the most effective fusion strategy can be empirical. However, our results presented in Table~\ref{table:overall_performance}, ~\ref{table:overall_performance_regression}, ~\ref{tab:intermediate-contrastive-loss-ablation}, \ref{tab:dude_auc_comparison}, and \ref{tab:performance_lit_pcba_transposed} provide strong evidence that our lightweight fusion strategy (early, intermediate, and late fusion) outperforms existing approaches in the literature. To guide the selection among these strategies, our intuition is as follows: if a modality is highly relevant to the downstream task, earlier fusion is likely to be more effective; otherwise, later fusion may be preferable.}

\textcolor{black}{
To test this hypothesis, we performed a retrospective analysis to assess the sensitivity of downstream tasks to different fusion strategies. Since early fusion embeddings often lack the flexibility to adapt to individual samples, we excluded them from this analysis. Instead, we used pretrained encoders to extract embeddings for each modality and performed a simple linear regression between the embeddings and task labels. We then computed the Pearson correlation between the predicted values and the ground truth as a measure of each modality's relevance.}

\textcolor{black}{
For each dataset, we recorded the highest correlation across all modalities as the "Top 1" score. We then concatenated the embeddings from all modalities and repeated the regression analysis. The improvement in correlation is reported as the "Pearson Gain." A higher Pearson Gain suggests that earlier fusion of multiple modalities is more beneficial.
}
\textcolor{black}{
As shown in Table~\ref{tab:fusion_strategies}, datasets where intermediate fusion performs best generally exhibit higher Pearson Gain compared to late fusion, supporting our intuition. However, for ESOL and FreeSolv, the correlation from a single modality is already high, making them less suitable for this analysis.
}

\begin{table}[t]
\scriptsize
\centering
{
\caption{Person correlation of different modalities and chosen fusion strategies across datasets.}
\label{tab:fusion_strategies}
\begin{tabular}{lcccccccccc}
\toprule
Dataset & Smiles & Image & NMR & FP & Peak & Top 1 & Concat & Pearson Gain & Strategy \\
\midrule
BBBP    & 0.757 & 0.765 & 0.733 & 0.746 & 0.484 & 0.765 & 0.935 & 0.170 & \textcolor{black}{Intermediate} \\
Bace    & 0.752 & 0.759 & 0.757 & 0.772 & 0.474 & 0.772 & 0.957 & 0.185 & \textcolor{black}{Intermediate} \\
Sider   & 0.598 & 0.579 & 0.578 & 0.592 & 0.403 & 0.597 & 0.973 & 0.376 & \textcolor{black}{Intermediate} \\
Hiv     & 0.305 & 0.295 & 0.239 & 0.287 & 0.193 & 0.305 & 0.420 & 0.115 & Late \\
MUV     & 0.159 & 0.159 & 0.141 & 0.147 & 0.050 & 0.159 & 0.315 & 0.156 & \textcolor{black}{Intermediate} \\
Clintox & 0.577 & 0.604 & 0.546 & 0.577 & 0.405 & 0.639 & 0.920 & 0.281 & \textcolor{black}{Intermediate} \\
tox21   & 0.550 & 0.558 & 0.578 & 0.565 & 0.183 & 0.578 & 0.691 & 0.113 & \textcolor{black}{Intermediate} \\
toxcast & 0.587 & 0.590 & 0.523 & 0.577 & 0.333 & 0.590 & 0.908 & 0.318 & \textcolor{black}{Intermediate} \\
Lipo    & 0.782 & 0.795 & 0.623 & 0.780 & 0.542 & 0.795 & 0.920 & 0.125 & Late \\
ESOL    & 0.958 & 0.960 & 0.893 & 0.947 & 0.705 & 0.960 & 0.999 & 0.039 & \textcolor{black}{Intermediate} \\
FreeSolv & 0.982 & 0.980 & 0.915 & 0.977 & 0.768 & 0.982 & 0.999 & 0.017 & \textcolor{black}{Intermediate} \\
\bottomrule
\end{tabular}
}
\end{table}


\section{Conclusion}

In summary, we introduce a novel relational learning metric for molecular graph representation that enhances the understanding of inter-instance relationships by capturing both local and global contexts. This is the first implementation of such a generalized metric in molecular graphs.Our method transforms pairwise self-similarity into relative similarity through a weighting function, allowing for complex relational insights. This metric is integrated into a multimodal representation, improving performance by utilizing modalities not directly accessible during fine-tuning. Empirical results show that our approach, MMFRL, excels in various molecular property prediction tasks. We also demonstrate detailed study about the explainability of the learned representations, offering valuable insights for drug molecule design. Despite these accomplishments, further exploration is needed to achieve more effective integration of graph- and node-level similarities. Looking ahead, we are enthusiastic about the prospect of applying our model to additional fields, such as social science, thereby broadening its applicability and impact.

\section{Dataset}
\subsection{Selected Modalities for Target Similarity Calculation}
The following modalities are used for target similarity calculation. For details on training the corresponding encoders to obtain fixed embeddings for these modalities, please refer to Appendix Section ~\ref{app:prefix-encoders}.

\textbf{Fingerprint:} Fingerprints are binary vectors that represent molecular structures, capturing the presence or absence of particular substructures, fragments, or chemical features within a molecule. \textcolor{black}{In particular, we utilize Morgan fingerprints, which are based on the Extended-Connectivity Fingerprints (ECFP) method introduced by Rogers and Hahn ~\citep{rogers2010ecfp}. Specifically, we generate fingerprints using AllChem.GetMorganFingerprintAsBitVect(mol, 2), which corresponds to ECFP4 (radius = 2). Because ECFP4 is one the most effective and Interpretable Molecular Representations ~\citep{zhong2023countmorgan}}.

\textbf{SMILES (Simplified Molecular Input Line Entry System):} SMILES offers a compact textual representation of chemical structures.

\textbf{NMR (Nuclear Magnetic Resonance):} NMR spectroscopy provides detailed insights into the chemical environment of atoms within a molecule ~\citep{bunzel2006nmr}. By analyzing the interactions of atomic nuclei with an applied magnetic field, NMR can reveal information about the structure, dynamics, and interactions of molecules, including the connectivity of atoms, functional groups, and conformational changes. In our experiments, $\text{NMR}_\text{spectrum}$ provides the information about the overal information of molecule while $\text{NMR}_\text{peak}$ provides the information about the individual atoms in the molecule.

\textbf{Image:} Images (e.g., 2D chemical structures) provide a visual representation of molecular structures. 

All of the similarity calculation from the modalities above are listed in Appendix ~\ref{appendix:uni-modal-self}.

\subsection{Pre-training}
NMRShiftDB-2 ~\citep{landrum2006rdkit} is a comprehensive database dedicated to nuclear magnetic resonance (NMR) chemical shift data, providing researchers with an extensive collection of expert-annotated NMR data for various organic compounds with molecular structures (SMILES). There are around 25,000 molecules used for pre-training and no overlap with downstream task datasets. And molecular images and graphs are generated via RDkit ~\citep{RDKit}.

\subsection{Downstream tasks}
Our model was trained on 11 drug discovery-related benchmarks sourced from MoleculeNet ~\citep{wu2018moleculenet}. Eight of them were classification tasks, including BBBP, BACE, SIDER, CLINTOX, HIV, MUV, TOX21, and ToxCas. The other three are regression tasks, including ESOL, Freesolv, and Lipo. Each dataset was divided into the train, validation, and test subsets in an 80\%:10\%:10\% ratio using the scaffold splitter ~\citep{halgren1996merck, landrum2006rdkit} from Chemprop ~\citep{yang2019analyzing, heid2023chemprop}. The scaffold splitter categorizes molecular data based on substructures, ensuring diverse structures in each subset. Molecules are partitioned into bins, with those exceeding half of the test set size assigned to training, promoting scaffold diversity in validation and test sets. Remaining bins are randomly allocated until reaching the desired set sizes, creating multiple scaffold splits for comprehensive evaluation. 

The DUD-E (Directory of Useful Decoys: Enhanced) dataset ~\citep{mysinger2012} is a widely used benchmark for virtual screening, containing 102 protein targets, thousands of active compounds, and carefully selected decoys that resemble actives in physico-chemical properties but differ topologically. In contrast, LIT-PCBA (Low-Throughput Informatics-Targeted PubChem BioAssay) ~\citep{trannguyen2020} offers a more realistic and challenging benchmark, derived from real experimental assays across 15 targets, with no artificial decoys and inherent data noise and imbalance. Together, they represent two ends of the spectrum in virtual screening evaluation—DUD-E with idealized conditions, and LIT-PCBA with real-world complexity. For the fine-tuning setting, We follow the same split and test approach as ~\citep{gao2023drugclip} for DUD-E and ~\citep{cai2022fpgnn} for LIT-PCBA.

\section{Methods}
We first explain the preliminaries, and then our proposed modified metric in relational learning to facilitate smooth alignment between graph and referred unimodality. Then, we introduce approaches for integrating multi modalities at different stages of the learning process.

\subsection{Molecular representation with DMPNN}
\label{Preliminary:all}
The Message Passing Neural Network (MPNN) ~\citep{gilmer2017neural} is a GNN model that processes an undirected graph $G$ with node (atom) features $x_v$ and edge (chemical bond) features $e_{vw}$. It operates through two distinct phases: a message passing phase, facilitating information transmission across the molecule to construct a neural representation, and a readout phase, utilizing the final representation to make predictions regarding properties of interest. The primary distinction between DMPNN and a generic MPNN lies in the message passing phase. While MPNN uses messages associated with nodes, DMPNN crucially differs by employing messages associated with directed edges ~\citep{yang2019analyzing}. This design choice is motivated by the necessity to prevent totters ~\citep{mahe2004extensions}, eliminating messages passed along paths of the form $v_1 v_2 \dots v_n$, where $v_i = v_{i+2}$ for some $i$, thereby eliminating unnecessary loops in the message passing trajectory.  

\subsection{Modified Relational Learning in pretraining}
Original Relation Learning ~\citep{zheng2021ressl} ensures that different augmented views of the same instance from computer vision tasks share similar features, while allowing for some variability. Suppose $z_i$ is the original embdding for the $i-th$ instance. Then $z_i^1$ is the embedding of first augmented view for $z_i$, and $z_i^2$ is the embedding of second augmented view for $z_i$. In this case, the Loss of Relational Learning (RL) is formulated as following:
\begin{equation*}
s_{ik}^1 = \frac{\mathbbm{1}_{i \neq k} \cdot \exp(z_i^1 \cdot z_k^2 / \tau)}{\sum_{j=1}^{N} \mathbbm{1}_{i \neq j} \cdot \exp(z_i^1 \cdot z_j^2 / \tau)}
\end{equation*}
\begin{equation*}
s_{ik}^2 = \frac{\mathbbm{1}_{i \neq k} \cdot \exp(z_i^2 \cdot z_k^2 / \tau_m)}{\sum_{j=1}^{N} \mathbbm{1}_{i \neq j} \cdot \exp(z_i^2 \cdot z_j^2 / \tau_m)}
\end{equation*}
\begin{equation*}
L_{RL} = -\frac{1}{N} \sum_{i=1}^{N} \sum_{\substack{k=1 \\ k \neq i}}^{N} s_{ik}^2 \log(s_{ik}^1).
\end{equation*}

We propose a modified relational metric by adapting the softmax function as a pairwise weighting mechanism. Let $|\mathcal{S}|$ denote the size of the instance set. The variable ${s}_{i,j}$ represents the learned similarity where $z_i$ is the embedding to be trained. On the other hand, ${t}_{i,j}^{R}$ defines the target similarity that captures the relationship between the pair of instances in the given space or modality $R$, where $z_i^{R}$ is a fixed embedding. The detailed formulation for the Loss of Modified Relatioal Learning (MRL) is provided below:
\begin{equation}
s_{i,j} = \frac{\exp(sim(z_i,z_j))}{\sum_{k=1}^{|\mathcal{S}|} \exp(sim(z_i,z_k))}
\end{equation}
\begin{equation}
t_{i,j}^R = \frac{\exp(sim(z_i^{R},z_j^{R}))}{\sum_{j=1}^{|\mathcal{S}|} \exp(sim(z_i^R,z_k^R))}
\end{equation}
\begin{equation}
\label{equ:lgscl}
L_{MRL} = -\frac{1}{|\mathcal{S}|} \sum_{i=1}^{|\mathcal{S}|} \sum_{\substack{j=1}}^{|\mathcal{S}|} t_{i,j}^R \log(s_{i,j}).
\end{equation}
Notably, unlike other similarity learning approaches ~\citep{wang2019multi, zhang2021jointly}, our method does not rely on the categorization of negative and positive pairs for the pair weighting function. Additionally, our use of the softmax function ensures that the generalized target similarity \( t_{i,j} \) adheres to the principles of convergence, which results in better ranking consistency between the graph modality and the auxiliary modality, compared with the original Relational Study, as follows:

\begin{theorem}[Convergence of Modified Relational Learning Metric]
\label{thm:convergent-smi-learning}
Let $\mathcal{S}$ be a set of instances with size of $|\mathcal{S}|$, and let $\mathcal{P}$ represent the learnable latent representations of instances in $\mathcal{S}$ such that $|\mathcal{P}| = |\mathcal{S}|$. For any two instances $i, j \in \mathcal{S}$, their respective latent representations are denoted by $\mathcal{P}_{i}$ and $\mathcal{P}_{j}$. Let $t_{i,j}$ represent the target similarity between instances $i$ and $j$ in a given domain, and let  $d_{i,j}$ be the similarity between $\mathcal{P}_i$ and $\mathcal{P}_j$ in the latent space. If $t_{i,j}$ is non-negative and $\{t_{i,j}\}$ satisfies the constraint $\sum_{j=1}^{|\mathcal{S}|}t_{i,j} = 1$, consider the loss function for an instance $i$ defined as follows:
\begin{equation}
    L(i) = -\sum_{j=1}^{|\mathcal{S}|} t_{i,j} \log \left( \frac{e^{d_{i,j}}}{\sum_{k=1}^{|\mathcal{S}|} e^{d_{i,k}}} \right)
\end{equation}
then when it reaches ideal optimum, the relationship between $t_{i,j}$ and $d_{i,j}$ satisfies:
\begin{equation}
    \text{softmax}(d_{i,j}) = t_{i,j}
\end{equation}
\end{theorem}
For detailed proof, please refer to Appendix Section ~\ref{appendix:gml-guide-proof}. 


\subsection{Fusion of multi-modality information in downstream tasks.}
During pre-training, the encoders are initialized with parameters derived from distinct reference modalities. A critical question that arises is how to effectively utilize these pre-trained models during the fine-tuning stage to improve performance on downstream tasks.

\subsubsection{Early Stage: Multimodal Multi-Similarity}
With a set of known target similarity $\{t^{R}\}$ from various modalities, we can transform themto multimodal space through a fusion function. There are numerous potential designs of the fusion function. For simplicity, we take linear combination as a demonstration. The multimodal generalized multi-similarity $t_{i,j}^{M}$ between $i^{th}$ and $j^{th}$ objects can be defined as follows: 
\begin{align}
  t_{i,j}^{M} &= fusion(\{t^{R}\})\\
  &= \sum w_{R} \cdot t_{i,j}^{R}
\label{equ:graph-guide-zhou-general}
\end{align}
where $t_{i,j}^{R}$ represents the target similarity between $i^{th}$ and $j^{th}$ instance in unimodal space $R$, $w_{R}$ is the pre-defined weights for the corresponding modal, and $\sum w_{R} = 1$. Then we can make $t_{i,j} = t_{i,j}^{R}$ in equation ~\ref{equ:lgscl}. Such that, it still satisfy the requirement of convergence (See proof in Appendix Section\ref{appendix:guarantee-fusion-sum}). In this case, the learnt similarity during pretraining will be aligned with this new combined target similarity.

\subsubsection{Intermediate Stage: Embedding concatenation and fusion}
\label{sec:intermediate-fusion}
Intermediate fusion integrates features from various modalities after their individual encoding processes and prior to the decoding/readout stage. Let \( \mathbf{f}_1, \mathbf{f}_2, \ldots, \mathbf{f}_n \) represent the feature vectors obtained from these different modalities. The resulting fused feature vector can be defined as follows:

\begin{equation}
\mathbf{f}_{\text{fused}} = \text{MLP}(\text{concat}(\mathbf{f}_1, \mathbf{f}_2, \ldots, \mathbf{f}_n))
\end{equation}

Where concat represents concatenation of the feature vectors. The fused features are then fed into a later readout function or decoder for downstrean tasks prediction or classification. The MLP (Multi-Layer Perceptron) is used to reduce the dimension to be the same as $\mathbf{f}_i$.

\subsubsection{Late Stage: decision-level}
\label{sec:late-fusion}
Late fusion (or decision-level fusion) combines the outputs of models trained on different modalities after they have been processed independently. Each modality is first processed separately, and their predictions are combined at a later stage.

Let \( p_1, p_2, \ldots, p_n \) be the predictions (e.g., probabilities) from different modalities. The final prediction \( p_{\text{final}} \) can be computed using a weighted sum mechanism:

\begin{equation}
w_i = T_i(\text{f}_i)
\end{equation}
\begin{equation}
p_i = \text{readout}_i(\text{f}_i)
\end{equation}
\begin{equation}
\label{equ:late-fusion-formula}
p_{\text{final}} = \sum_{i=1}^{n} w_i p_i
\end{equation}

Where \( w_i \) are the weights assigned to each modality's prediction, and they can be adjusted based on the importance of each modality. In particular, $w_i$ is tunable during the learning process for respective downsteak tasks.

\textcolor{black}{
\section*{Data Availability}
The pretraining data can be downloaded from \href{https://nmrshiftdb.nmr.uni-koeln.de/}{NMRShiftDB2}. The MoleculeNet dataset is available at \href{https://moleculenet.org/}{MoleculeNet}. The DuD-E dataset can be accessed at \href{http://dude.docking.org/}{DuD-E}, and the Lit-PCBA dataset can be downloaded from \href{https://drugdesign.unistra.fr/LIT-PCBA/}{Lit-PCBA}.
}
\textcolor{black}{
\section*{Code Availability}
The code is available in Github: \href{https://github.com/zhengyjo/MMRFL}{https://github.com/zhengyjo/MMRFL}
}



\bibliography{example_paper_clean_revised}
\bibliographystyle{iclr2025_conference}

\newpage
\appendix
\onecolumn
\begin{center}
    \Large{Appendix}
\end{center}
\counterwithin{figure}{section}
\counterwithin{equation}{section}
\counterwithin{table}{section}
\setcounter{figure}{0} 
\setcounter{equation}{0} 
\setcounter{table}{0} 
\section {Multi-Similarity \& Contrastive Learning}
\subsection{Multi-Similarities in Contrastive Learning}
Two distinct types of similarities, as illustrated in Appendix Figure \ref{fig:similarities}, can be identified: \textit{self-similarity} (the pairwise similarity between two objects, typically defined through cosine similarity) and \textit{relative similarity} (distinctions in self-similarity with other pairs) ~\citep{wang2019multi}.
\begin{figure}[H] 
    \centering
    \includegraphics[width=0.48\textwidth]{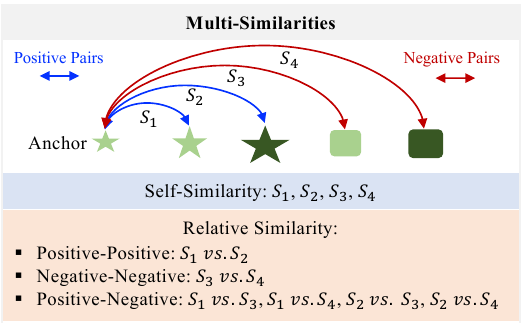}
    \caption{Illustration of Different Types of Similarities.}
    \label{fig:similarities}
\end{figure}

\subsection{Current Molecular Graph Contrastive Learning Approaches}
In current molecular graph contrastive learning approaches, positive pairs are commonly formed through either \textit{data augmentation} ~\citep{sun2021mocl, you2020graph}, employing techniques such as node deletion, edge perturbation, subgraph extraction, attribute masking, and subgraph substitution, or \textit{domain knowledge}, as demonstrated by reactant-product pairing ~\citep{wang2022chemicalreactionaware} or conformer grouping ~\citep{moon20233d}.

\begin{figure}[H] 
    \centering
    \includegraphics[width=0.9\textwidth]{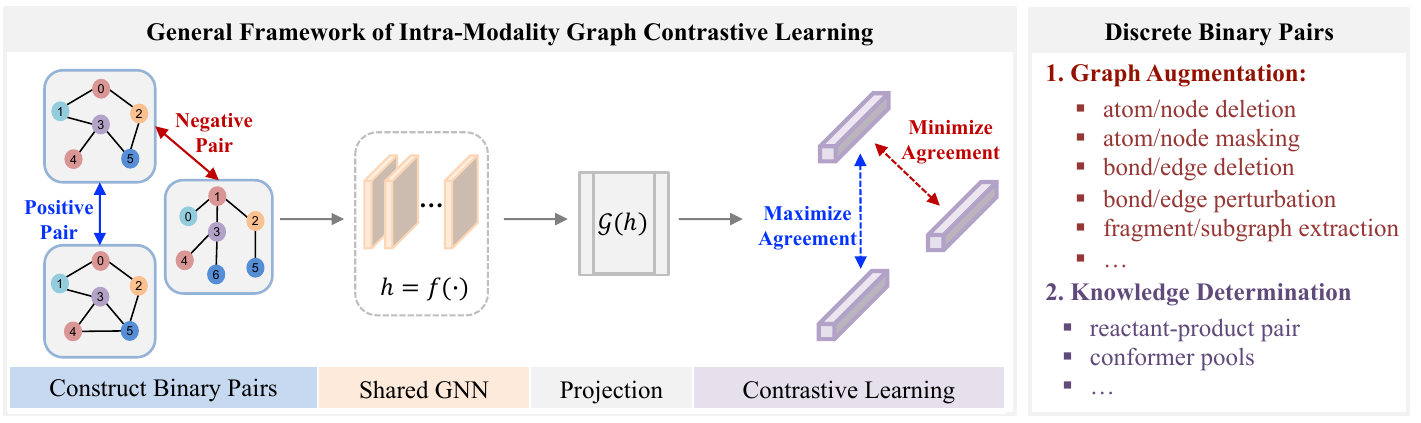}
    \caption{General framework of Intra-Modality Graph Contrastive Learning. It relies on definition of positive and negative pairs.}
    \label{fig:traditional-cl}
\end{figure}

\section{Supplementary Proof}
\subsection{Revisiting Theorem of Convergent Similarity Learning}
\label{appendix:gml-guide-proof}
Let $\mathcal{S}$ be a set of instances with size $|\mathcal{S}|$, and let $\mathcal{P}$ represent the tunable latent representations of instances in $\mathcal{S}$ such that $|\mathcal{P}| = |\mathcal{S}|$. For any two instances $i, j \in \mathcal{S}$, their latent representations are denoted by $\mathcal{P}_{i}$ and $\mathcal{P}_{j}$, respectively. Let $t_{i,j}$ represent the target similarity between instances $i$ and $j$ in a given domain, and $d_{i,j}$ be the similarity between $\mathcal{P}_i$ and $\mathcal{P}_j$ in the latent space.
\begin{theorem}[Theorem of Convergent Similarity learning]
Given $t_{i,j}$ is non-negative and $\{t_{i,j}\}$ satisfies the constraint $\sum_{j=1}^{|\mathcal{S}|}t_{i,j} = 1$, consider the loss function for an instance $i$ defined as follows:
\begin{equation}
    L(i) = -\sum_{j=1}^{|\mathcal{S}|} t_{i,j} \log \left( \frac{e^{d_{i,j}}}{\sum_{k=1}^{|\mathcal{S}|} e^{d_{i,k}}} \right)
\end{equation}
then when it reaches ideal optimum, the relationship between $t_{i,j}$ and $d_{i,j}$ satisfies:
\begin{equation}
    \text{softmax}(d_{i,j}) = t_{i,j}
\end{equation}
\end{theorem}

\begin{proof}
In order to optimize the loss $L(i)$, we need to set the following partial derivative to be 0 for each $d_{i,j}$ with $1\leq j \leq |\mathcal{M}|$. Here are the detailed steps:
\begin{align*}
\frac{\partial L(i)}{\partial {d_{i,j}}} 
&= \frac{\partial}{\partial d_{i,j}}\underbrace{\left( - t_{i,j} \log \frac{e^{d_{i,j}}}{e^{d_{i,j}} + \sum_{\substack{k \neq j}} e^{d_{i,k}}} \right)}_{\text{When the numerator includes } e^{d_{i,j}}} + \frac{\partial}{\partial d_{i,j}}\underbrace{\left( \sum_{\substack{k \neq j}} - t_{i,k} \log \frac{e^{d_{i,k}}}{e^{d_{i,j}} + \sum_{\substack{k \neq j}} e^{d_{i,k}}} \right)}_{\text{When the numerator does not include }e^{d_{i,j}}} \\
&= -(t_{i,j} - t_{i,j} \cdot \text{softmax}(d_{i,j})) - \sum_{\substack{k \neq j}} t_{i,k} \cdot \text{softmax}(d_{i,j}) \\
&= - \left( t_{i,j} - \left(t_{i,j} + \sum_{\substack{k \neq j}} t_{i,k} \right) \cdot \text{softmax}(d_{i,j})\right)
\end{align*}
Since $\sum_{l=1}^{|\mathcal{M}|}t_{i,l} = 1$, we can further simplify it as 
\begin{align*}
\frac{\partial L(i)}{\partial {d_{i,j}}} = - (t_{i,j} -  \text{softmax}(d_{i,j}))
\end{align*}
In order to optimize, we need to see the above partial derivative to be 0:
\begin{align*}
\frac{\partial L(i)}{\partial {d_{i,j}}} = - (t_{i,j} -  \text{softmax}(d_{i,j})) = 0
\end{align*}
In addition, the corresponding second partial derivative denoted as $\frac{\partial L(i)}{\partial {d^2_{i,j}}}$ manifests as follows:
\begin{align*}
\frac{\partial L(i)}{\partial {d^2_{i,j}}} = \text{softmax}(d_{i,j})(1- \text{softmax}(d_{i,j}))
\end{align*}
As $\text{softmax}(d_{i,j})$ takes values within the open interval (0,1), it follows that $\frac{\partial L(i)}{\partial {d^2_{i,j}}}$ is always positive. Consequently, the global optimum is global minimum.\\
Furthermore, when it comes to optimum:
\begin{align*}
t_{i,j} &= \text{softmax}(d_{i,j}) \\
d_{i,j} &= \log(t_{i,j}) + \log \left( \sum_{\substack{1\leq l \leq |\mathcal{M}|}} e^{d_{i,j}} \right)
\end{align*}
It is easy to show that when it reaches optimum, $d_{i,j}$ is consistent with target similarity metric $t_{i,j}$. Without loss of generosity, suppose $t_{i,j} > t_{i,j'}$ :
\begin{align*}
d_{i,j} - d_{i,j'} &= \log(t_{i,j}) + \log \left( \sum_{\substack{1\leq l \leq |\mathcal{M}|}} e^{d_{il}} \right) - \left( \log(t_{i,j'}) + \log \left( \sum_{\substack{1\leq l \leq |\mathcal{M}|}} e^{d_{il}} \right) \right) \\
&= \log(t_{i,j}) - \log(t_{i,j'}) \\
&= \log\left( \frac{t_{i,j}}{t_{i,j'}} \right) > 0
\end{align*}
\end{proof}

\subsection{Guarantee of Sum of Fused Multimodal Similarity}
\label{appendix:guarantee-fusion-sum}
Given sets of uni-modal generalized similarity $\{t^{R}\}$ and $\sum w_{t^{R}} = 1$, the sum of fused multimodal similarity also equals 1, as demonstrated below:
\begin{align*}
\sum (t_{i,j}^{R})
& = \sum \sum(w_{R} \cdot t^{R}_{i,j})\\
& = \sum (w_{R}  \sum t^{R}_{i,j})\\
& = \sum w_{R} \cdot 1 = 1
\label{equ:graph-guide-zhou-sum}
\end{align*}

\section{Revisiting Target Similarity Settings}
\subsection{Encoders \& Packages}
\label{app:prefix-encoders}
To derive the target similarities, we need to reply on pre-trained encoders or well-defined packages as follows \textcolor{black}{in Table~\ref{tab:encoder_similarity}}:

\begin{table}[h]
\label{tab:encoder_similarity}
\begin{center}
\begin{small}
\caption{Encoders and packages used to produce self-similarities}
\begin{tabular}{l|c|c|c}
\hline
Unimodal & Representation & Encoder/Package & Pre-trained Source \\ 
\hline
Image & 2D image & CNN & Img2mol ~\citep{clevert2021img2mol} \\
SMILES & Sequence &Transformer  & CReSS ~\citep{yang2021cross}  \\
\textsuperscript{13}CNMR Spectrum& Sequence & 1D CNN & AutoEncoder ~\citep{costanti2023deep}  \\
\textsuperscript{13}CNMR peak & Scalar & NMRShiftDB2 ~\citep{steinbeck2003nmrshiftdb} & N/A  \\
Fingerprint & Sequence & RDKit ~\citep{landrum2006rdkit}  & N/A  \\
\hline
\end{tabular}
\end{small}
\end{center}
\end{table}

\textcolor{black}{We selected GIN ~\cite{xu2018powerful} as the graph encoder for Smile, Image, Fingerprint, NMR, and NMR-Peak, respectively. \textcolor{black}{All modalities share a consistent structure, each with 5 layers, an embedding dimension of 128, and a projection dimension of 512.} In addition, during pretraining, contrastive learning is performed on each modality independently. For instance, if molecule A only possesses SMILE, IMAGE, and Fingerprint data, but lacks NMR information, it will be included in the training for contrastive learning on SMILE, IMAGE, and Fingerprint, but not on NMR. In contrast, for early fusion, all modalities must be present for the included molecules.}



\subsection{Target Similarity at Graph Level}
\label{appendix:uni-modal-self}
\textbf{Fingerprint.} The mathematical formula of fingerprint similarity, denoted as $S_{i,j}^{F}$, can be viewed as follows:
\begin{align}
   S_{i,j}^{F} & = Tanimoto(A, B) = \frac{|A \cap B|}{|A \cup B|}  
\end{align}
where \( A \) and \( B \) are sets of molecular fragments for molecule $i$ and $j$, and \( |A \cap B| \) and \( |A \cup B| \) denote the size of their intersection and union, respectively.

\textbf{Image.} The self-similarity for Image, denoted as $S_{i,j}^{I}$, can be defined as follows:
\begin{equation}
    S_{i,j}^{I} = Cos( \mathcal V_{i}, \mathcal V_{j}) = \frac{\mathcal{V}_i \cdot \mathcal{V}_j^T}{\|\mathcal{V}_i\| \cdot \|\mathcal{V}_j\|}
\end{equation}
where $\mathcal V_{i},  \mathcal V_{j} $ represents the embedding of Image for two given molecules.

\textbf{NMR Spectrum.} The self-similarity for NMR spectrum, denoted as $S_{i,j}^{C}$, can be defined as follows:
\begin{equation}
    S_{i,j}^{C} = Cos( \mathcal V_{i}, \mathcal V_{j}) = \frac{\mathcal{V}_i \cdot \mathcal{V}_j^T}{\|\mathcal{V}_i\| \cdot \|\mathcal{V}_j\|}
\end{equation}
where $\mathcal V_{i},  \mathcal V_{j} $ represents the embedding of NMR spectra for two given molecules.

\textbf{Smiles.} The self-similarity for Smiles, denoted as $S_{i,j}^{S}$, can be defined as follows:
\begin{equation}
    S_{i,j}^{S} = Cos( \mathcal V_{i}, \mathcal V_{j}) = \frac{\mathcal{V}_i \cdot \mathcal{V}_j^T}{\|\mathcal{V}_i\| \cdot \|\mathcal{V}_j\|}
\end{equation}
where $\mathcal V_{i},  \mathcal V_{j} $ represents the embedding of Smiles for two given molecules.

\textbf{NMR Peak} The similarity among nodes (atoms) is derived from the positions of their signal peaks on $^{13}$C NMR spectra, measured in parts per million (ppm). The ppm values are continuous, typically ranging from 0 to 200 (see more introduction of ppm in Appendix \ref{appendix:ppm}). The self-similarity of NMR peaks $S^{P}_{l,m}$ can be defined as following:
\begin{equation}
    S^{P}_{l,m} = \frac{\tau_{2}}{|ppm_{l} - ppm_{m}|+\tau_{1}}
\end{equation}
where $ppm_{l}$ and $ppm_{m}$ are the positions of NMR peaks for the $l^{th}$, $m^{th}$ Carbon atom, $\tau_{1}$ and $\tau_{2}$ are temperature hyper-parameter. \textcolor{black}{Additionally, we conducted an ablation study to examine the impact of different temperature combinations ($\tau_1$ and $\tau_2$) on peak self-similarity and its effect on model performance for \textit{Atom Alignment with Peak Accuracy}. Our findings indicate that the following combinations do not significantly affect performance, as shown in Table~\ref{tab:tau_ablation}. For this analysis, we fixed the GIN depth at 5, set the GIN embedding dimensionality to 128, and kept the projection dimension at 512. The results suggest that the best performance occurs when $\tau_1 = 10^{-5}$ and $\tau_2 = 10^1$.}

\begin{table}[h]
\centering
{
\caption{Ablation study of Effect of different $\tau_1$ and $\tau_2$ combinations on accuracy.}
\label{tab:tau_ablation}
\begin{tabular}{ccc}
\toprule
\textbf{$\tau_1$} & \textbf{$\tau_2$} & \textbf{Accuracy (\%)} \\
\midrule
$10^{-1}$ & $10^{1}$ & 89.6 \\
$10^{-1}$ & $10^{2}$ & 89.8 \\
$10^{-1}$ & $10^{3}$ & 89.6 \\
$10^{-1}$ & $10^{4}$ & 88.9 \\
$10^{-1}$ & $10^{5}$ & 89.3 \\
$10^{-2}$ & $10^{1}$ & 89.8 \\
$10^{-2}$ & $10^{2}$ & 89.8 \\
$10^{-2}$ & $10^{3}$ & 88.8 \\
$10^{-2}$ & $10^{4}$ & 87.2 \\
$10^{-2}$ & $10^{5}$ & 89.4 \\
$10^{-3}$ & $10^{1}$ & 89.2 \\
$10^{-3}$ & $10^{2}$ & 89.1 \\
$10^{-3}$ & $10^{3}$ & 89.0 \\
$10^{-3}$ & $10^{4}$ & 89.7 \\
$10^{-3}$ & $10^{5}$ & 89.4 \\
$10^{-4}$ & $10^{1}$ & 89.8 \\
$10^{-4}$ & $10^{2}$ & 89.7 \\
$10^{-4}$ & $10^{3}$ & 89.8 \\
$10^{-4}$ & $10^{4}$ & 89.5 \\
$10^{-4}$ & $10^{5}$ & 88.4 \\
\textbf{$10^{-5}$} & \textbf{$10^{1}$} & \textbf{90.0} \\
$10^{-5}$ & $10^{2}$ & 89.5 \\
$10^{-5}$ & $10^{3}$ & 89.6 \\
$10^{-5}$ & $10^{4}$ & 89.7 \\
$10^{-5}$ & $10^{5}$ & 89.7 \\
\bottomrule
\end{tabular}
}
\end{table}

\subsection{A Brief Introduction to PPM for NMR Peak}
\label{appendix:ppm}
In chemistry, $^{13}$C NMR stands out as a common technique for structural analysis by revealing molecular structures by elucidating the chemical environments of carbon atoms and their magnetic responses to external fields ~\citep{gerothanassis2002nuclear,lambert2019nuclear}. It quantifies these features in parts per million (ppm) relative to a reference compound, such as tetramethylsilane (TMS), thereby simplifying comparisons across experiments. As a result, the continuous peak positions, measured in parts per million (ppm), offer a robust knowledge span—a natural ordering metric that can be employed to derive measures of similarity ~\citep{xu2023molecular}.

\subsection{Configuration of Early Fusion}
\label{appendix:config-fusion}
A simple linear combination is used to formulate the multimodal relational similarity $t_{i,j}^{M}$ between the $i^{th}$ and $j^{th}$ molecules, represented as as follows:
\begin{align}
t_{i,j}^{M} = w_{SM} \cdot t^{SM}_{i,j} + w_{{C}} \cdot t^{C}_{i,j} + w_{{I}} \cdot t^{I}_{i,j} + w_{{F}} \cdot t^{F}_{i,j} + w_{{F}} \cdot t^{F}_{i,j} + w_{{P}} \cdot t^{P}_{i,j}
\label{equ:graph-guide-zhou}
\end{align}
where $t^{SM}_{i,j}$ denotes the similarity based on SMILES, $t^{C}_{i,j}$ denotes the similarity with respect to $^{13}$C NMR spectrum, $t^{I}_{i,j}$ denotes the similarity regarding images, $F$ denotes the similarity based on fingerprints, and $P$ denotes the similarity based on fingerprints. $w_{SM}$, $w_{C}$, $w_{I}$, and $w_{F}$ are the pre-defined weights for their respective similarity, and $w_{SM} + w_{{C}} + w_{{I}} + w_{{F}} + w_{{P}} = 1$. The ablation study about how the weight combinations influence the performance of early fusion is shown in Table~\ref{tab:modality_weights}

\begin{table*}[h]
\centering
\tiny
\setlength{\tabcolsep}{2pt}
\caption{\zzy{Performance of different modality weight combinations across datasets in early fusion.}}
\label{tab:modality_weights}
\zzy{
\begin{tabular}{ccccc|ccccccccccc}
\toprule
\multicolumn{5}{c|}{\textbf{Modality Weight}} & \multicolumn{11}{c}{\textbf{Dataset}} \\
\textbf{Smiles} & \textbf{NMR} & \textbf{Image} & \textbf{FP} & \textbf{Peak} & \textbf{BBBP} & \textbf{BACE} & \textbf{SIDER} & \textbf{CLINTOX} & \textbf{HIV} & \textbf{MUV} & \textbf{TOX21} & \textbf{TOXCAST} & \textbf{ESOL} & \textbf{FREESOLV} & \textbf{LIPO} \\
\midrule
1 & 0   & 0   & 0   & 0 & {92.9$\pm$1.5} & 90.9$\pm$3.3 & 64.9$\pm$0.3 & 78.2$\pm$1.9  & \textbf{83.3$\pm$1.1} & 80.1$\pm$2.5 & \underline{85.7$\pm$1.2} &70.5$\pm$2.5 & 0.811$\pm$ 0.109 & \underline{1.623$\pm$ 0.168} & \underline{0.539$\pm$ 0.017}\\
0 & 1   & 0   & 0   & 0  & 91.0$\pm$2.0 & \underline{93.2$\pm$2.7} & \textbf{68.1$\pm$1.5} & \textbf{87.7$\pm$6.5}  & 80.9$\pm$5.0 & \underline{80.9$\pm$5.0} &  85.1$\pm$0.4 & \textbf{71.1$\pm$0.8} & 0.844$\pm$ 0.123 & 2.417$\pm$ 0.495 & 0.609$\pm$ 0.031 \\
0 & 0   & 1   & 0   & 0 & \underline{93.1$\pm$2.4} & {92.9$\pm$1.8} & 65.3$\pm$1.5 & 86.2$\pm$6.5 & \underline{82.3$\pm$0.6} & 78.7$\pm$1.7 & \textbf{86.0$\pm$1.0} & \underline{71.0$\pm$1.6} & \textbf{0.761$\pm$ 0.068} & 1.648$\pm$ 0.045 & \textbf{0.537$\pm$ 0.005} \\
0 & 0   & 0   & 1   & 0 & 92.9$\pm$2.3 & 91.7$\pm$3.6  & {65.6$\pm$0.7} & \underline{87.5$\pm$6.0} & 81.2$\pm$2.5 & \textbf{82.9$\pm$3.1} & 85.3$\pm$1.3 &70.0$\pm$1.4 & \underline{0.808$\pm$ 0.071} & \textbf{1.437$\pm$ 0.134} & 0.565$\pm$ 0.017\\
0 & 0   & 0   & 0   & 1 & \textbf{93.4$\pm$2.7} & 89.3$\pm$1.7 & 62.8$\pm$2.1 & 86.1$\pm$5.4& 82.1$\pm$0.4 &75.4$\pm$5.2 &84.9$\pm$1.0 & 70.6$\pm$0.8 & 0.924$\pm$0.083 &1.707$\pm$0.126 & 0.587$\pm$0.021\\
0.2 & 0.2 & 0.2 & 0.2 & 0.2 & 91.6$\pm$5.0 & \textbf{94.3$\pm$2.4} & \underline{66.4$\pm$1.9} & 85.3$\pm$6.8 & 82.0$\pm$4.2 & 80.6$\pm$3.2 & 85.2$\pm$0.6 & 69.8$\pm$1.1 & 1.037$\pm$0.090 & 2.093$\pm$0.090 & 0.607$\pm$0.034 \\
\bottomrule
\end{tabular}
}
\end{table*}

\section{Experimental Settings}
\label{appendix:exp-setting}
\subsection{Pre-Training Setting}
During pretraining, we utilized an Adam optimizer with a learning rate set to 0.001, spanning 200 epochs and employing a batch size of 256. The model was trained on around 25,000 data points. The NMR data were experimental data, extracted from NMRShiftDB2 ~\citep{steinbeck2003nmrshiftdb}. Other chemical modalities, such as images, fingerprints and graphs, were produced from SMILES by RDKit ~\citet{landrum2006rdkit}.


\subsection{Fine-Tuning Setting}

\subsubsection{Datasets} 

For fine-tuning, our model was trained on 11 drug discovery-related benchmarks sourced from MoleculeNet ~\citep{wu2018moleculenet}. Eight of these benchmarks were designated for classification downstream tasks, including BBBP, BACE, SIDER, CLINTOX, HIV, MUV, TOX21, and ToxCast, while three were allocated for regression tasks, namely ESOL, Freesolv, and Lipo. The datasets were divided into train/validation/test sets using a ratio of 80\%:10\%:10\%, accomplished through the scaffold splitter ~\citep{halgren1996merck, landrum2006rdkit} from Chemprop ~\citep{yang2019analyzing, heid2023chemprop}, like previous works. The scaffold splitter categorizes molecular data based on substructures, ensuring diverse structures in each set. Molecules are partitioned into bins, with those exceeding half of the test set size assigned to training, promoting scaffold diversity in validation and test sets. Remaining bins are randomly allocated until reaching the desired set sizes, creating multiple scaffold splits for comprehensive evaluation.


\subsubsection{Baselines}
\label{appendix:baselines}
We systematically compared MMFRL's performance with various state-of-the-art baseline models across different categories. In the realm of supervised models, AttentiveFP ~\citep{xiong2019pushing} and DMPNN ~\citep{yang2019analyzing} stand out by leveraging graph attention networks and node-edge interactive message passing, respectively. The unsupervised learning method N-Gram ~\citep{liu2019n} employs graph embeddings and short walks for graph representation. Predictive self-supervised learning methods, such as GEM ~\citep{fang2022geometry} and Uni-Mol ~\citep{Zhou2023UniMolAU}, are specifically designed for predicting molecular geometric information. Moreover, our evaluation encompasses a range of contrastive learning methods, namely InfoGraph ~\citep{sun2019infograph}, GraphCL ~\citep{you2020graph}, MolCLR ~\citep{wang2022molecular}, and GraphMVP ~\citep{liu2022pretraining}, all serving as essential baselines. The baseline results are collected from recent works ~\citep{fang2022geometry, Zhou2023UniMolAU, moon20233d, fang2023knowledge}.

\subsubsection{Evaluation} 
To assess the effectiveness of our fine-tuned model, we measure the ROC-AUC for classification downstream tasks, and the root mean squared error (RMSE) metric for regression tasks. In order to ensure a fair and robust comparisons, we conduct three independent runs using three different random seeds for scaffold splitting across all datasets. The reported performance metrics are then averaged across these runs, and the standard deviation is computed as prior works. In paritcular, the random selected seeds for respective experiments are drawn from the range between 0 and 20.

\end{document}